\documentclass[11pt]{article}
\usepackage[margin=1in]{geometry}
\usepackage{amsmath,amssymb,amsthm,mathtools}
\usepackage{booktabs}
\usepackage{graphicx}
\usepackage{microtype}
\usepackage{enumitem}
\usepackage{needspace}
\usepackage{url}
\usepackage[hidelinks]{hyperref}
\newcommand{\doi}[1]{\href{https://doi.org/#1}{\nolinkurl{doi:#1}}}
\graphicspath{{figures/}}

\newtheorem{theorem}{Theorem}[section]
\newtheorem{lemma}[theorem]{Lemma}
\newtheorem{proposition}[theorem]{Proposition}
\newtheorem{corollary}[theorem]{Corollary}
\newtheorem{remark}[theorem]{Remark}
\newtheorem{definition}[theorem]{Definition}

\newcommand{\MMS}{\operatorname{MMS}}
\newcommand{\PMMS}{\operatorname{PMMS}}
\newcommand{\EFX}{\operatorname{EFX}}
\newcommand{\OPT}{\operatorname{OPT}}
\newcommand{\PoA}{\operatorname{PoA}}

\providecommand{\Description}[1]{}
\title{Tight Maximin Share Guarantees of EFX and PMMS}
\author{%
  Qinghua Qin$^{1,2,3}$\thanks{Email: \texttt{qinqinghua231@mails.ucas.ac.cn},\texttt{tsinghuaqin@pku.org.cn}.}\\[4pt]
  \small $^{1}$University of Chinese Academy of Sciences \quad $^{2}$Chinese Academy of Sciences\\
  \small $^{3}$Huawei Technologies Co., Ltd.
}
\date{August 2026}
\hypersetup{pdftitle={Tight Maximin Share Guarantees of EFX and PMMS},pdfauthor={Qinghua Qin}}

\begin{document}
\maketitle

\begin{abstract}
In the fair division of indivisible goods among autonomous agents, allocations frequently satisfy strong pairwise fairness criteria while failing to guarantee an agent's full maximin share. Maximin share (MMS) is a global benchmark: the value an agent can guarantee by partitioning all goods into $n$ bundles and receiving her least preferred bundle. In contrast, local pairwise fairness notions---most notably envy-freeness up to any good (EFX) and pairwise maximin share (PMMS)---evaluate fairness through bilateral bundle comparisons. We resolve the open question of determining the optimal universal MMS approximation factor guaranteed by either local notion under nonnegative additive valuations, proving that both factors equal exactly $10/17$.

For the lower bound, we show that both EFX and PMMS imply a common local residual inequality: the value of any foreign bundle after removing its least valuable item cannot exceed the focal agent's allocated value. We design a continuous piecewise weight function that converts this local constraint into the global $10/17$ MMS guarantee. The same charging argument establishes that $\alpha$-EFX implies $(10\alpha/17)$-MMS. For tightness, we construct an explicit family of complete allocations that simultaneously satisfy PMMS and EFX$_0$, including comparisons involving zero-valued goods, and whose MMS ratios converge to $10/17$. For each fixed number of agents, we also establish exact reciprocal identities connecting fair division to scheduling: the EFX guarantee is the reciprocal of the pure price of anarchy in selfish identical-machine covering, while the PMMS guarantee is the reciprocal of a locality gap for exact repartition of two machines.
\end{abstract}

\section{Introduction}

The fair allocation of indivisible resources among autonomous agents is a
cornerstone problem in multi-agent systems and algorithmic economics
\cite{AmanatidisSurvey2023}. Fairness is traditionally formalized through two
complementary paradigms. Share-based guarantees compare an agent's allocated
bundle against what she could secure via a global partition of the entire item
pool. The preeminent benchmark in this category is the maximin share (MMS): an
agent envisions partitioning all goods into $n$ bundles and evaluates the
guarantee by her least preferred bundle \cite{Budish2011}. Because exact MMS
allocations do not always exist even under additive valuations
\cite{KurokawaEtAl2018}, considerable effort has focused on approximation
algorithms \cite{amanatidis2015mms,ghodsi2018fair,garg2021mms}.

In contrast, envy-based criteria are local and bilateral. Envy-freeness up to
any good (EFX) requires that no agent envy another after the hypothetical
removal of any positively valued good from the peer's bundle. Pairwise maximin
share (PMMS) asks whether an agent already obtains the value she could secure by
optimally bipartitioning the union of her own bundle and any single peer bundle.
Unlike the global $n$-way partition problem inherent to MMS, these conditions rely
solely on pairwise bundle comparisons. EFX is directly verifiable from item
values, whereas checking PMMS entails solving an exact two-way partition
subproblem. A fundamental theoretical question emerges: how much of the global
MMS benchmark does every allocation satisfying either local pairwise criterion
inherently guarantee?

Amanatidis, Birmpas, and Markakis~\cite[Propositions 3.4 and 4.3]{ABM2018}
established that EFX guarantees a $4/7$ fraction of MMS and gave finite
constructions with ratios below $0.5914$. The implication from PMMS to EFX was
already known~\cite[Theorem 4.6]{CaragiannisEtAl2019}. Section~\ref{sec:sylvester}
extends the historical finite pattern to arbitrary depth and computes its
limit. Closing the remaining gap has stood as a known open problem. In this
work, we resolve this question by settling the exact universal factor for
both notions:
\begin{equation}\label{eq:mainconstant}
\rho^{\EFX\to\MMS}=\rho^{\PMMS\to\MMS}=\frac{10}{17}.
\end{equation}
Each factor represents the infimum of the ratio between an agent's allocated
value and her positive MMS, taken over all complete allocations satisfying the
specified local criterion across all numbers of agents. Consequently,
PMMS---though strictly stronger as a local requirement---provides no better
worst-case universal MMS protection than EFX.

The universal lower bound stems from a common local inequality. Consider an
arbitrary evaluating agent, normalize her allocated bundle's value to one, and
prune zero-valued goods as well as empty or singleton peer bundles. By
Lemma~\ref{lem:reduction}, this reduction leaves the agent's MMS non-decreasing.
Crucially, both pairwise notions force every remaining peer bundle $B$ to obey
\begin{equation}\label{eq:localconditionintro}
v(B)-\min_{g\in B}v(g)\le1.
\end{equation}
Adapting the weight-based charging framework introduced for machine covering by
Chen et al.~\cite{ChenEtAl2013}, we design a direct fair-division potential
argument. We prove that any remaining peer bundle satisfying~\eqref{eq:localconditionintro}
carries a weight of at most one, whereas any bundle achieving a value of at
least $17/10$ must carry a weight of at least one. Since the evaluating agent's
own bundle contributes strictly less than one unit of weight, any global
partition where every bundle exceeds $17/10$ would exceed the total available
weight, establishing the $10/17$ guarantee. The same charging scheme scales
naturally to show that $\alpha$-EFX implies $(10\alpha/17)$-MMS.

Tightness requires constructing an allocation that simultaneously satisfies the
local constraints across all peer bundles while admitting a global repartition
that certifies an MMS close to $17/10$. We build on the combinatorial block
arithmetic of Chen et al.~\cite{ChenEtAl2013} and introduce a seven-block
symmetric absorption construction. This gadget isolates a singleton of value one
for the evaluating agent while cleanly absorbing all residual items into peer
bundles. The resulting complete allocations satisfy both PMMS and EFX$_0$
(even for goods of zero value) across all agents, with their MMS ratios
converging to $10/17$. Section~\ref{sec:tight} details this construction and
verifies its structural properties.

For any fixed number of agents, the connection to scheduling is exact and
structural. The extremal EFX-to-MMS ratio is precisely the reciprocal of the
pure price of anarchy for selfish identical-machine covering. Similarly, the
PMMS-to-MMS ratio is the reciprocal of a locality gap comparing a designated
machine against the exact optimal repartition of its jobs with those on each
peer machine. We prove these fixed-$n$ identities in Section~\ref{sec:scheduling}.
Finally, Section~\ref{sec:sylvester} analyzes the structural barrier in earlier
constructions, showing that the specified homogeneous layered model limits
recursive certificates to a sub-optimal limit strictly above $10/17$.

\paragraph{Contributions and provenance.}
We determine the exact universal factors, extend the bound to $\alpha$-EFX, and
establish the fixed-$n$ reciprocal dualities. The piecewise weight function and
base block arithmetic originate in the machine-covering analysis
of~\cite{ChenEtAl2013}. Our conceptual contributions lie in formulating the
unifying local fair-division reduction, developing the seven-block symmetric
absorption construction that achieves complete PMMS and EFX$_0$ allocations,
and characterizing the structural limitations of the specified homogeneous
layered model.

\paragraph{Related work.}
Much of the MMS literature investigates existence thresholds and approximation
algorithms \cite{KurokawaEtAl2018,amanatidis2015mms,ghodsi2018fair,garg2021mms}.
Our inquiry instead examines the MMS protection guaranteed by established
pairwise envy-based criteria. Envy-graph techniques yield EF1
allocations~\cite{lipton2004fair}; Nash welfare maximization, with the original
zero-welfare convention, yields EF1 and Pareto optimality~\cite[Definition 3.1
and Theorem 3.2]{CaragiannisEtAl2019}. EFX remains significantly more elusive
under general valuations \cite{plaut2020envy}. The quantitative comparison
between local and global fairness was initiated in~\cite{ABM2018}; recent work
explores epistemic variants of PMMS \cite{FeldmanEtAl2026}.

\begin{figure}[t]
\centering
\includegraphics[width=.78\textwidth]{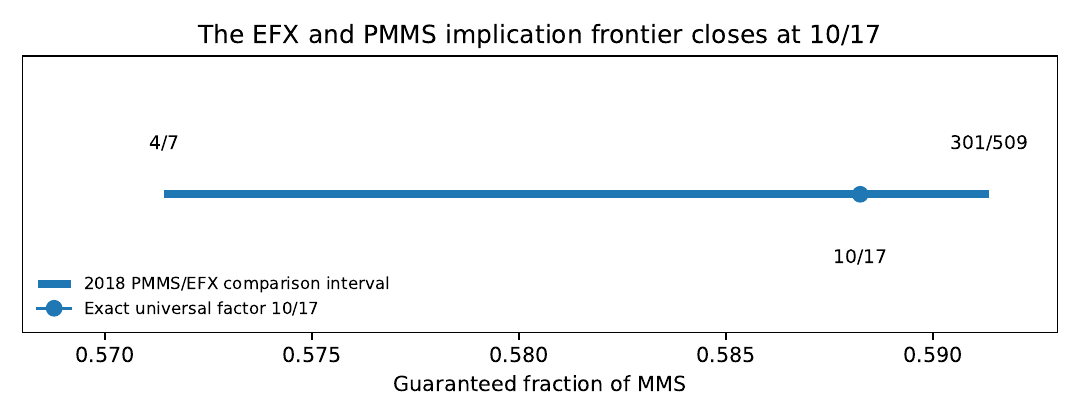}
\caption{The historical interval from $4/7$ to the finite-depth upper certificate
$301/509$ obtained from the construction in~\cite[Proposition 4.3]{ABM2018}, and
the sharp universal factor $10/17$. Continuing the historical recursion slightly
lowers its upper endpoint to $c_\infty$, but leaves a strict gap to $10/17$.}
\Description{An interval from four sevenths to three hundred one over five
hundred nine, with the exact factor ten seventeenths marked inside it.}
\label{fig:frontier}
\end{figure}

In scheduling, Chen et al.~\cite{ChenEtAl2013} proved that the universal pure
price of anarchy for identical-machine covering is $17/10$. Our fixed-$n$
reciprocal identities provide an economic explanation for why this constant
governs fair division. While the reduction from EFX to scheduling uses dummy
padding, its reverse mapping is exact and preserves ratios without distortion.

\section{Preliminaries and Reduction}\label{sec:prelim}

Let $N=[n]$ denote a set of $n$ agents and let $M$ be a finite set of indivisible
goods. Each agent $i\in N$ possesses a nonnegative additive valuation $v_i$,
meaning that
\[
v_i(S)=\sum_{g\in S}v_i(g)\qquad(S\subseteq M).
\]
An allocation $A=(A_1,\ldots,A_n)$ is a partition of $M$ into $n$ disjoint bundles.

\begin{definition}[EFX and EFX$_0$]
An allocation $A$ is \emph{envy-free up to any good} (EFX) if for every pair of
distinct agents $i,j$ and every good $g\in A_j$ with $v_i(g)>0$,
\[
v_i(A_i)\ge v_i(A_j\setminus\{g\}).
\]
The allocation is EFX$_0$ if the inequality holds for every $g\in A_j$, including
goods of zero value to $i$. For $\alpha\in(0,1]$, an allocation is $\alpha$-EFX if
\[
v_i(A_i)\ge\alpha\,v_i(A_j\setminus\{g\})
\]
for every positive-valued good $g\in A_j$ and every $j\ne i$.
\end{definition}

\begin{definition}[MMS and PMMS]
Let $\Pi_k(S)$ denote the collection of all partitions of a subset $S\subseteq M$
into $k$ possibly empty bundles. For an integer $k\ge1$, define
\[
\mu_i(k,S)=\max_{(P_1,\ldots,P_k)\in\Pi_k(S)}\min_{r\in[k]}v_i(P_r).
\]
The maximin share of agent $i$ is
\[
\MMS_i:=\mu_i(n,M).
\]
An allocation is $\beta$-MMS if $v_i(A_i)\ge\beta\MMS_i$ for every agent $i$.
An allocation satisfies \emph{pairwise maximin share} (PMMS) if for every pair of
distinct agents $i,j$,
\[
v_i(A_i)\ge\mu_i(2,A_i\cup A_j).
\]
\end{definition}

When evaluating fairness from the perspective of an arbitrary agent $i$, we refer
to $A_j$ ($j\ne i$) as a peer bundle. We write $s=v_i(A_i)$ for agent $i$'s
allocated value and denote her normalized valuation by $v$. The construction
depth $\ell$ in Section~\ref{sec:tight} is conceptually distinct from the machine
load notations $\ell_r(S)$ in Section~\ref{sec:scheduling}.

For any fairness notion $\mathcal F\in\{\EFX,\EFX_0,\PMMS\}$, we define the
worst-case $n$-agent approximation ratio as
\[
\rho_n^{\mathcal F}
=\inf\left\{\frac{v_i(A_i)}{\MMS_i}:
 A\text{ is }\mathcal F,\ |N|=n,\ \MMS_i>0\right\},
\]
where the infimum is taken over all instances with nonnegative additive
valuations and complete allocations. The universal approximation factor across
all numbers of agents is $\rho^{\mathcal F\to\MMS}=\inf_{n\ge2}\rho_n^{\mathcal F}$,
with $\rho_1^{\mathcal F}=1$. These factors quantify the intrinsic protection
guaranteed by any allocation satisfying the designated local criterion,
independently of whether such allocations always exist.

Eliminating goods with $v_i(g)=0$ preserves agent $i$'s $k$-way maximin share
$\mu_i(k,M)$ for all $k$, as zero-valued items can be added to or removed from any
bundle without altering its perceived value. This simplification also preserves
all bilateral EFX and PMMS constraints from agent $i$'s perspective; no
preservation claim for other agents is needed.

The following reduction formalizes that eliminating a good or an empty allocated
bundle cannot decrease an agent's benchmark~\cite[Lemma 3.2]{ABM2018}.

\begin{lemma}[Pruning goods and agents]\label{lem:reduction}
For any nonnegative additive valuation $v_i$, any good $g\in M$, and any $n\ge2$,
\[
\mu_i(n-1,M\setminus\{g\})\ge\mu_i(n,M).
\]
Similarly, removing an empty allocated bundle along with its nonfocal recipient
cannot decrease agent $i$'s maximin share.
\end{lemma}
\begin{proof}
Let $(T_1,\ldots,T_n)$ be an optimal $n$-partition achieving $\MMS_i$, and assume
without loss of generality that $g\in T_n$. Removing $g$ and distributing the
remaining goods in $T_n\setminus\{g\}$ arbitrarily among the first $n-1$ bundles
yields an $(n-1)$-partition of $M\setminus\{g\}$. Since no remaining bundle loses
value, the new minimum bundle value is at least the original MMS. If an agent
receives an empty bundle, merging two parts of an MMS partition immediately
yields the same conclusion.
\end{proof}

\section{A Common Local Inequality and the 10/17 Bound}\label{sec:direct}

Our lower-bound proof uses the weight function introduced for selfish machine
covering by Chen et al.~\cite[author manuscript, Section 2.1]{ChenEtAl2013}.
We establish its properties directly in the context of fair division: any
remaining peer bundle satisfying~\eqref{eq:localconditionintro} has total weight
at most one, whereas any bundle worth at least $17/10$ carries weight at least
one. The breakpoints $1/2$ and $2/3$ naturally capture structural transitions in
the packing of locally constrained bundles, while the tight multiset
$\{1,1/2,1/5\}$ dictates the threshold $17/10$.

Define the weight function $w:[0,1]\to[0,1/2]$ by
\begin{equation}\label{eq:weight}
w(x)=
\begin{cases}
\dfrac{x}{1+x},&0\le x\le\dfrac12,\\[6pt]
\dfrac{x}{2-x},&\dfrac12<x\le\dfrac23,\\[6pt]
\dfrac12,&\dfrac23<x\le1.
\end{cases}
\end{equation}
The function $w$ is continuous, nondecreasing, concave on $[0,1/2]$, and convex
on $[1/2,2/3]$. For all $x>0$, $w(x)<x$, and on the interval $[0,2/3]$, it
satisfies the linear lower bound
\begin{equation}\label{eq:two-thirds-weight}
w(x)\ge\frac23x.
\end{equation}

\begin{figure}[t]
\centering
\includegraphics[width=.70\textwidth]{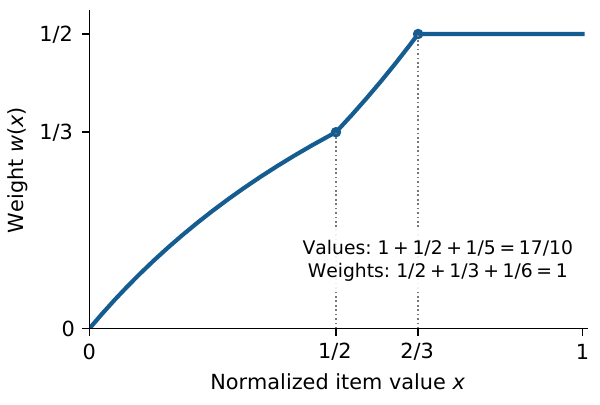}
\caption{The weight function with breakpoints at $1/2$ and $2/3$. The critical
set $\{1,1/2,1/5\}$ has value $17/10$ and weight one. A normalized local bundle
has weight at most one; an MMS part worth at least $17/10$ has weight at least one.}
\Description{A continuous piecewise weight curve on item values zero to one,
with slope changes at one half and two thirds. The values one, one half, and
one fifth sum to seventeen tenths and their weights sum to one.}
\label{fig:charging}
\end{figure}

\begin{lemma}[Weight of a locally constrained bundle]\label{lem:bundleweight}
Let $B$ be a finite multiset of at least two positive numbers in $(0,1]$. If
\begin{equation}\label{eq:residualbundle}
\sum_{x\in B}x-\min B\le1,
\end{equation}
then
\[
\sum_{x\in B}w(x)\le1.
\]
\end{lemma}
\begin{proof}
Let $p=\min B$.

\emph{Case 1: $p\le1/3$.}
We show that every element $x\in B$ satisfies
\begin{equation}\label{eq:w-bound-p}
w(x)\le\frac{x}{1+p}.
\end{equation}
Indeed, if $x\le1/2$, then $x\ge p$ immediately gives
$w(x)=x/(1+x)\le x/(1+p)$. If $1/2<x\le2/3$, then $x\le2/3$ and $p\le1/3$
imply $x+p\le1$, from which $2-x\ge1+p$ and
$w(x)=x/(2-x)\le x/(1+p)$ follow. If $x>2/3$, then
$2x\ge4/3\ge1+p$, ensuring $w(x)=1/2\le x/(1+p)$.
Summing~\eqref{eq:w-bound-p} over all $x\in B$ and
applying~\eqref{eq:residualbundle} yields
\[
\sum_{x\in B}w(x)\le\frac{\sum_{x\in B}x}{1+p}\le1.
\]

\emph{Case 2: $p>1/2$.}
Because every element exceeds $1/2$, condition~\eqref{eq:residualbundle}
restricts $B$ to containing exactly two elements; three such items would have a
sum minus minimum exceeding one. Since $w(x)\le1/2$ for all $x$, the sum of their
weights is at most one.

\emph{Case 3: $1/3<p\le1/2$.}
Here $B$ can contain at most three elements. The claim is immediate if $|B|=2$.
When $|B|=3$, write $B=\{p,x,y\}$ with $p\le x\le y$.
Condition~\eqref{eq:residualbundle} implies $x+y\le1$. If $y\le1/2$, all three
elements fall in the first branch, so each weight is at most $1/3$, summing to
at most one. If $y>1/2$, then necessarily $x<1/2$ and $y<2/3$. Since
$y\le1-x$ and $x>1/3$, we have $1-x\in(1/2,2/3)$, which falls into the middle
branch of $w$. By monotonicity of $w$,
\[
w(y)\le w(1-x)=\frac{1-x}{2-(1-x)}=\frac{1-x}{1+x}.
\]
Therefore,
\[
w(x)+w(y)\le\frac{x}{1+x}+\frac{1-x}{1+x}
=\frac1{1+x}\le\frac1{1+p}.
\]
Adding $w(p)=p/(1+p)$ completes the proof.
\end{proof}

The charging threshold $17/10$ arises from the dual lower bound on weight.

\begin{lemma}[Weight of a high-value bundle]\label{lem:highweight}
Let $S$ be a finite multiset of numbers in $[0,1]$. If
\[
\sum_{x\in S}x\ge\frac{17}{10},
\]
then
\[
\sum_{x\in S}w(x)\ge1.
\]
\end{lemma}
\begin{proof}
If every $x\in S$ satisfies $x\le2/3$, the linear
bound~\eqref{eq:two-thirds-weight} implies
\[
\sum_{x\in S}w(x)\ge\frac23\cdot\frac{17}{10}=\frac{17}{15}>1.
\]
If $S$ contains at least two elements exceeding $2/3$, those two elements alone
contribute a combined weight of $1/2+1/2=1$. It remains to analyze the case where
exactly one element $z\in S$ lies in $(2/3,1]$. This element contributes
$w(z)=1/2$, and the remaining elements sum to at least $17/10-z\ge7/10$.

It therefore suffices to show that any multiset of values $x_r\in[0,2/3]$
summing to at least $7/10$ carries total weight at least $1/2$. By monotonicity,
we may reduce their sum to exactly $7/10$. Let $q=\max_r x_r$.

If $q\le1/2$, all elements lie on $f(x)=x/(1+x)$. Since $f$ is strictly concave
on $[0,1/2]$, minimizing the separable sum $\sum_r f(x_r)$ over $[0,1/2]$
subject to the fixed total value is a concave minimization over a compact convex
polytope. An extreme point attains the minimum; at such a point at most one
entry lies strictly in $(0,1/2)$ while all others belong to $\{0,1/2\}$.
With a total sum of $7/10$, the positive entries must be $1/2$ and $1/5$,
giving total weight
\[
\frac{1/2}{1+1/2}+\frac{1/5}{1+1/5}=\frac13+\frac16=\frac12.
\]

If instead $q\in(1/2,2/3]$, let $t=7/10-q<1/5$. The remaining elements sum to
$t$ and all lie on the first branch. For any $a,b\ge0$,
\[
\frac{a}{1+a}+\frac{b}{1+b}-\frac{a+b}{1+a+b}
=\frac{ab(2+a+b)}{(1+a)(1+b)(1+a+b)}\ge0.
\]
Since $t<1/5$, merging the remaining elements keeps their sum on the first
branch and cannot increase their weight. Thus,
\[
\sum_r w(x_r)\ge\frac{q}{2-q}+\frac{7/10-q}{17/10-q}.
\]
Subtracting $1/2$ yields
\[
\frac{q}{2-q}+\frac{7/10-q}{17/10-q}-\frac12
=\frac{(2q-1)(6-5q)}{2(2-q)(17-10q)}\ge0,
\]
which is nonnegative for all $q\in[1/2,2/3]$, completing the argument.
\end{proof}

\begin{remark}[Tightness of the charging threshold]
The multiset $\{1,1/2,1/5\}$ achieves a sum of $17/10$ with weight
\[
\frac12+\frac13+\frac16=1.
\]
Reducing the item $1/5$ to $1/5-\varepsilon$ for any $\varepsilon\in(0,1/5)$
results in a total value of $17/10-\varepsilon$ and a weight strictly below one.
Hence, the threshold $17/10$ is mathematically optimal for the weight
function~\eqref{eq:weight}.
\end{remark}

\begin{lemma}[Local-to-global charging]\label{lem:charging}
Let $v$ be a positive additive valuation on a finite set $M$ of goods partitioned
into $k\ge1$ bundles: an evaluating agent's bundle $A_c$ and $k-1$ peer bundles.
Suppose $0<v(A_c)\le1$, and every peer bundle $B$ contains at least two goods and
satisfies $v(B)-\min_{g\in B}v(g)\le1$. Then
\[
\mu(k,M)<\frac{17}{10}.
\]
\end{lemma}
\begin{proof}
Every item in $A_c$ has value at most $v(A_c)\le1$. For any item $g$ in a peer
bundle $B$, choosing an alternate item $h\in B\setminus\{g\}$ yields
\[
v(g)\le v(B\setminus\{h\})\le v(B)-\min_{q\in B}v(q)\le1.
\]
Thus every item has value in $[0,1]$ and belongs to the domain of $w$.
By Lemma~\ref{lem:bundleweight}, every peer bundle carries weight at most one.
The evaluating agent's bundle $A_c$ is nonempty and satisfies $w(x)<x$ for all
$x>0$, so its weight is strictly less than $v(A_c)\le1$. The total weight across
all allocated items is therefore strictly less than $k$. If the $k$-way maximin
share were at least $17/10$, an optimal $k$-partition would contain $k$ bundles
each worth at least $17/10$, which by Lemma~\ref{lem:highweight} would demand a
total weight of at least $k$, a contradiction. This also covers $k=1$.
\end{proof}

\subsection{EFX implies 10/17-MMS}

\begin{theorem}[EFX guarantee]\label{thm:efx-direct}
Under nonnegative additive valuations, every EFX allocation is a $(10/17)$-MMS
allocation.
\end{theorem}
\begin{proof}
Consider an arbitrary agent $i$ and let $s=v_i(A_i)$. If $s=0$, EFX requires
that every peer bundle contain at most one good of positive value to agent $i$;
otherwise, dropping one positive good would still leave positive value,
violating EFX. The item set therefore contains at most $n-1$ goods of positive
value, implying that any $n$-partition has a zero-value bundle, whence $\MMS_i=0$.

Assume $s>0$. We delete all goods with $v_i(g)=0$. Next, we iteratively remove
any peer agent whose bundle is empty or a singleton, along with the singleton
item when present. By Lemma~\ref{lem:reduction}, these prunings preserve the
bilateral EFX conditions from agent $i$'s perspective and cannot decrease her
maximin share. If this reduction terminates with $n'=1$, agent $i$ receives all
remaining goods and $s=\mu_i(1,M')\ge\MMS_i$, proving the guarantee. Otherwise,
the reduced instance has $n'\ge2$ agents, every peer bundle contains at least
two goods of positive value, and its unscaled MMS satisfies $\mu_i'\ge\MMS_i$.

Normalizing $v_i$ by $1/s$ gives $v(A_i)=1$. For every peer bundle $B$ and every
good $g\in B$, EFX ensures $v(B\setminus\{g\})\le1$, so in particular
\[
v(B)-\min_{g\in B}v(g)\le1.
\]
Applying Lemma~\ref{lem:charging} to this normalized reduced instance yields
\[
\frac{\MMS_i}{s}\le\frac{\mu_i'}{s}<\frac{17}{10},
\]
which establishes $s\ge(10/17)\MMS_i$.
\end{proof}

\subsection{PMMS implies 10/17-MMS directly}

PMMS enforces the exact same local condition through an explicit bipartition.

\Needspace{10\baselineskip}
\begin{lemma}[Local inequality enforced by PMMS]\label{lem:pmms-local}
Let $v$ be an additive valuation and let $A_i,B$ be two allocated bundles such
that $v(A_i)=1$. If $B$ contains at least two goods, all of positive value, and
\[
\mu(2,A_i\cup B)\le1,
\]
then
\[
v(B)-\min_{g\in B}v(g)\le1.
\]
\end{lemma}
\begin{proof}
Choose $g\in\arg\min_{h\in B}v(h)$. If $v(B\setminus\{g\})>1$, the bipartition
of $A_i\cup B$ into $B\setminus\{g\}$ and $A_i\cup\{g\}$ produces two bundles
each with value strictly exceeding one, since $v(g)>0$. This implies
$\mu(2,A_i\cup B)>1$, contradicting the hypothesis.
\end{proof}

\begin{theorem}[PMMS guarantee]\label{thm:pmms-direct}
Under nonnegative additive valuations, every PMMS allocation is a $(10/17)$-MMS
allocation.
\end{theorem}
\begin{proof}
Fix an evaluating agent $i$. If $v_i(A_i)=0$, no peer bundle can contain two
goods of positive value: separating one positive good from the remaining
positive goods in a bipartition of $A_i\cup B$ would give positive value on both
sides, violating PMMS. Hence $\MMS_i=0$.

Assuming $s=v_i(A_i)>0$, we remove zero-valued goods and iteratively prune
singleton or empty peer bundles as in Theorem~\ref{thm:efx-direct}. If $n'=1$,
agent $i$ receives all remaining goods and $s=\mu_i(1,M')\ge\MMS_i$. Otherwise,
for $n'\ge2$, normalizing so that $v(A_i)=1$, PMMS ensures that for every
remaining peer bundle $B$,
\[
\mu(2,A_i\cup B)\le1.
\]
Lemma~\ref{lem:pmms-local} therefore implies $v(B)-\min_{g\in B}v(g)\le1$.
Applying Lemma~\ref{lem:charging} yields a normalized reduced maximin share
strictly below $17/10$, proving $s\ge(10/17)\MMS_i$.
\end{proof}

\subsection{Approximate EFX}

\begin{corollary}[$\alpha$-EFX]\label{thm:alpha}
For every $\alpha\in(0,1]$, every $\alpha$-EFX allocation under nonnegative
additive valuations is a $(10\alpha/17)$-MMS allocation.
\end{corollary}
\begin{proof}
If $v_i(A_i)=0$, the positive-good argument in Theorem~\ref{thm:efx-direct}
again guarantees $\MMS_i=0$. For $s=v_i(A_i)>0$, apply the same pruning reductions
and scale the valuation by $\alpha/s$. The evaluating agent's bundle now has
value $\alpha\in(0,1]$, and for every peer bundle $B$ and positive good $g\in B$,
the $\alpha$-EFX condition becomes
\[
v(B\setminus\{g\})\le1.
\]
Lemma~\ref{lem:charging} applies directly to this normalized reduced system,
yielding $\alpha\MMS_i/s\le\alpha\mu_i'/s<17/10$, from which the
$(10\alpha/17)$-MMS guarantee follows. No tightness claim for $\alpha<1$ is needed.
\end{proof}

\section{A Common Tight Family}\label{sec:tight}

To establish the matching upper bound, we construct an explicit family of
instances featuring two concurrent partitions of the same items: one satisfies
the local pairwise constraints~\eqref{eq:localconditionintro}, while the second
certifies an MMS benchmark approaching $17/10$. The underlying block arithmetic
adapts the construction of Chen et al.~\cite{ChenEtAl2013}.

For an integer depth $\ell\ge1$, define
\begin{equation}\label{eq:hi}
\begin{gathered}
h_i=\frac{2(4^i-1)}{3}\quad(i\ge0),\qquad m_\ell=2(10^\ell-1),\\
\delta_\ell=\frac1{30h_\ell}=\frac1{20(4^\ell-1)}.
\end{gathered}
\end{equation}
The parameter $m_\ell$ denotes the number of bundles per block, and
$\delta_\ell\to0$ as $\ell\to\infty$. The integers $h_i$ ensure that the item
values satisfy all required combinatorial identities.

For each index $i\in\{1,\ldots,\ell\}$, we define four positive item values:
\begin{align}
a_i&=\frac12+(h_i-1)\delta_\ell,
&b_i&=\frac12-(h_i-1)\delta_\ell,\label{eq:ab}\\
c_i&=\frac15+4h_{i-1}\delta_\ell,
&d_i&=\frac15-h_i\delta_\ell.\label{eq:cd}
\end{align}
A single \emph{block} contains $m_\ell$ unit-valued items together with, for each
$i=1,\ldots,\ell$:
\begin{itemize}[nosep]
\item $6\cdot10^{\ell-i}$ items of value $a_i$;
\item $12\cdot10^{\ell-i}$ items of value $b_i$;
\item $12\cdot10^{\ell-i}$ items of value $c_i$;
\item $6\cdot10^{\ell-i}$ items of value $d_i$.
\end{itemize}

Using the recurrence $h_i=4h_{i-1}+2$, these items satisfy the exact algebraic
relations:
\begin{equation}\label{eq:block-identities}
\begin{gathered}
a_i+b_i=1,\qquad c_{i+1}+4d_i=1\quad(i<\ell),\\
c_1=\frac15,\qquad d_\ell=\frac16,\\
a_i+d_i=b_i+c_i=\frac7{10}-\delta_\ell.
\end{gathered}
\end{equation}
Because $0=h_0<h_1<\cdots<h_\ell$ and $h_\ell\delta_\ell=1/30$, the four
perturbed item values are strictly between zero and one; all goods, including
the unit goods, have value at most one. In particular, $a_i\ge b_i$ and
$d_i\le1/5\le c_{i+1}$ for $i<\ell$, so the minimum item values within each
bundle are as detailed in Table~\ref{tab:localpartition}.

\subsection{Two partitions of one block}

The first partition arranges the goods within a block into $m_\ell$ bundles
satisfying the pairwise condition, as specified in Table~\ref{tab:localpartition}.

\begin{table}[t]
\centering
\small
\begin{tabular}{@{}lll@{}}
\toprule
Bundle type & Number of bundles & Value after deleting a minimum good\\
\midrule
$\{1,1\}$ & $m_\ell/2$ & $1$\\
$\{a_i,b_i,b_i\}$ & $6\cdot10^{\ell-i}$ for each $i$ & $a_i+b_i=1$\\
$\{c_1,c_1,c_1,c_1,c_1,c_1\}$ & $2\cdot10^{\ell-1}$ & $5c_1=1$\\
$\{c_{i+1},d_i,d_i,d_i,d_i,d_i\}$ & $12\cdot10^{\ell-i-1}$ for $i<\ell$ & $c_{i+1}+4d_i=1$\\
$\{d_\ell,d_\ell,d_\ell,d_\ell,d_\ell,d_\ell\}$ & $1$ & $5/6$\\
\bottomrule
\end{tabular}
\caption{A partition of one block into $m_\ell$ bundles. Every bundle has
residual value at most one after deletion of a minimum-valued good.}
\label{tab:localpartition}
\end{table}

This configuration uses each item in the block exactly once. Bundles of type
$\{a_i,b_i,b_i\}$ consume two $b_i$ goods per $a_i$ good; items of value $c_1$
are grouped into bundles of size six; and mixed bundles use
$5\cdot12\cdot10^{\ell-i-1}=6\cdot10^{\ell-i}$ goods of value $d_i$.
The unique terminal bundle contains the six remaining goods of value
$d_\ell=1/6$. The total bundle count matches
\begin{align*}
&\frac{m_\ell}{2}+\sum_{i=1}^\ell6\cdot10^{\ell-i}
 +2\cdot10^{\ell-1}+\sum_{i=1}^{\ell-1}12\cdot10^{\ell-i-1}+1\\
&=\frac{m_\ell}{2}+\frac23(10^\ell-1)+2\cdot10^{\ell-1}
 +\frac43(10^{\ell-1}-1)+1=m_\ell.
\end{align*}

The second partition establishes the high MMS certificate. Within each block,
we pair every good of value $a_i$ with one of value $d_i$, and every good of
value $b_i$ with one of value $c_i$. The total number of such pairs is
\[
\sum_{i=1}^{\ell}(6+12)10^{\ell-i}
=18\cdot\frac{10^\ell-1}{9}=m_\ell.
\]
By~\eqref{eq:block-identities}, each pair has value exactly $7/10-\delta_\ell$.
Combining each pair with one of the $m_\ell$ unit-valued goods produces
$m_\ell$ equal-value bundles, each having value
\begin{equation}\label{eq:Cell}
C_\ell:=\frac{17}{10}-\delta_\ell.
\end{equation}

\subsection{Seven-block absorption and the evaluating singleton}

A primary difficulty in translating machine-covering constructions into complete
fair-division allocations lies in enforcing the local fairness criteria across
all agents. In machine covering, low-load machines are passive; in fair
division, every agent must actively evaluate her bundle against all peers.

To resolve this, we exploit the fact that when an agent $i$ receives a singleton
bundle $\{z\}$ of value one, the PMMS condition against any nonempty peer bundle
$B$ of positive-valued goods simplifies to the residual EFX condition:
\begin{equation}\label{eq:singleton-pmms}
\mu(2,\{z\}\cup B)\le1
\quad\Longleftrightarrow\quad
v(B)-\min_{g\in B}v(g)\le1.
\end{equation}
Indeed, if the residual exceeds one, placing a minimum good with $z$ leaves both
parts above one. Conversely, in any bipartition, the side not containing $z$ is
either all of $B$, leaving the other side at value one, or a proper subset of
$B$, worth at most $v(B)-\min B\le1$. Comparisons with empty peer bundles are
automatic.

In the single-block partition (Table~\ref{tab:localpartition}), the terminal
bundle $\{d_\ell,\ldots,d_\ell\}$ consists of six goods of value $1/6$.
To grant the evaluating agent a singleton bundle $\{z\}$ of value one without
discarding goods, we combine seven identical blocks, labeled $0,1,\ldots,6$,
alongside an extra unit good $z$ through a symmetric absorption construction:
\begin{enumerate}[nosep]
\item \emph{Creating the evaluating bundle:} In block $0$, we extract all six
goods of value $d_\ell=1/6$ from its terminal bundle and replace them with the
unit good $z$, forming the bundle $\{z\}$.
\item \emph{Symmetric absorption:} The six displaced goods of value $1/6$ are
distributed across the other six blocks, adding exactly one item of value $1/6$
to the terminal bundle of each block $1,\ldots,6$.
\item \emph{Preserving local bounds:} The modified terminal bundle in each of
blocks $1,\ldots,6$ now contains seven goods of value $1/6$, yielding a total
value of $7/6$. Removing any single item leaves value $7/6-1/6=1$, satisfying
the local bound~\eqref{eq:residualbundle}.
\end{enumerate}
The complete instance consists of
\begin{equation}\label{eq:Nell}
N_\ell=7m_\ell=14(10^\ell-1)
\end{equation}
agents and bundles.

\begin{lemma}[Pairwise fairness of the constructed allocation]\label{lem:focal-tight}
Under the valuation defined above, the allocation satisfies both EFX$_0$ and
PMMS from the evaluating agent's perspective.
\end{lemma}
\begin{proof}
The evaluating agent receives the singleton $A_c=\{z\}$ of value one. Every
unchanged peer bundle in blocks $0,\ldots,6$ is listed in
Table~\ref{tab:localpartition} and satisfies $v(B)-\min_{g\in B}v(g)\le1$.
Each of the six augmented terminal bundles contains seven items of value $1/6$;
removing any item leaves value exactly one. Thus EFX$_0$ holds with respect to
every peer bundle. The singleton equivalence~\eqref{eq:singleton-pmms} then
guarantees that PMMS holds for all bilateral comparisons involving the
evaluating agent.
\end{proof}

To extend fairness to all agents, we define each peer agent's valuation such
that she values every item in her assigned bundle at one, and all other items
at zero. Under these valuations, every peer agent experiences zero envy and
satisfies PMMS against all other bundles: the entire two-bundle union has no
more value than her own bundle. The complete allocation is therefore
simultaneously EFX$_0$ and PMMS for all agents.

\Needspace{10\baselineskip}
\begin{theorem}[Tight family]\label{thm:tight-family}
For every $\ell\ge1$, there exists an instance with
$N_\ell=14(10^\ell-1)$ agents admitting a complete allocation that is
simultaneously PMMS and EFX$_0$, in which an evaluating agent receives value one
while her maximin share satisfies $\MMS_i\ge C_\ell=17/10-\delta_\ell$.
Consequently, her allocated-value/MMS ratio is upper bounded by
\begin{equation}\label{eq:tight-ratio}
r_\ell:=\frac1{C_\ell}
=\frac{20(4^\ell-1)}{34\cdot4^\ell-35}
\xrightarrow{\ell\to\infty}\frac{10}{17}.
\end{equation}
Moreover, $r_\ell$ strictly decreases to this limit.
\end{theorem}
\begin{proof}
The fairness properties follow from Lemma~\ref{lem:focal-tight} and the peer
valuations above. To establish the MMS lower bound, we regroup the total item
pool by its original block labels and use the second block configuration: each
of the seven blocks yields $m_\ell$ bundles of value exactly $C_\ell$, producing
$7m_\ell=N_\ell$ bundles across all blocks. The additional unit good $z$ can be
placed into any bundle without reducing any share. Hence $\MMS_i\ge C_\ell$.
Since $\delta_\ell$ strictly decreases to zero, $r_\ell$ strictly decreases to
$10/17$.
\end{proof}

Evaluating the certificate bounds $r_\ell$ across the first four construction
depths illustrates the rapid convergence to $10/17$:
\begin{center}
\small
\begin{tabular}{@{}cccc@{}}
\toprule
Depth $\ell$ & Agents $N_\ell$ & MMS certificate $C_\ell$ & Ratio upper bound $r_\ell$\\
\midrule
1 & 126 & $101/60$ & $60/101\approx0.594059$\\
2 & 1386 & $509/300$ & $300/509\approx0.589391$\\
3 & 13986 & $2141/1260$ & $1260/2141\approx0.588510$\\
4 & 139986 & $8669/5100$ & $5100/8669\approx0.588303$\\
\bottomrule
\end{tabular}
\end{center}
For $\ell=2$, the certificate already strictly improves upon the historical
recursive limit. In fact, the bound proved in Section~\ref{sec:sylvester} gives
\begin{equation}\label{eq:300509}
\frac{300}{509}<\frac{36}{61}\le c_\infty<\frac{301}{509}.
\end{equation}

\begin{corollary}[Exact universal constants]\label{cor:exact}
Under nonnegative additive valuations,
\[
\boxed{\rho^{\EFX\to\MMS}=\rho^{\PMMS\to\MMS}=\frac{10}{17}.}
\]
The same constant holds for EFX$_0$.
\end{corollary}
\begin{proof}
The lower bounds are established by Theorems~\ref{thm:efx-direct}
and~\ref{thm:pmms-direct}. For the upper bounds, let $q_\ell$ denote the true
ratio $v_i(A_i)/\MMS_i$ for the evaluating agent in
Theorem~\ref{thm:tight-family}. By construction,
$10/17\le q_\ell\le r_\ell$. Taking the limit as $\ell\to\infty$ squeezes
$q_\ell$ to $10/17$. Since EFX$_0$ implies EFX, the factor $10/17$ is exact
across all three notions.
\end{proof}

\section{Fixed-\texorpdfstring{$n$}{n} Correspondence with Machine Covering}\label{sec:scheduling}

For any fixed number of agents $n$, the extremal MMS approximation ratios mirror
two dual machine-covering formulations. We formalize this connection by relating
EFX to pure Nash equilibria and PMMS to an exact two-machine repartition gap.

\subsection{Identical-machine covering and pure Nash equilibria}

An instance of the identical-machine covering problem consists of $m$ jobs with
positive processing times $p_1,\ldots,p_m$ and $n$ identical machines. A schedule
$S$ assigns each job to a machine. Letting $\ell_r(S)$ denote the load on machine
$r$, the covering objective evaluates the minimum machine load:
\[
\ell_{\min}(S)=\min_{r\in[n]}\ell_r(S).
\]
The optimal covering value is $\OPT_n(P)=\max_S\ell_{\min}(S)$
\cite{woeginger1997cover}. In the associated selfish machine-covering game,
each job is an agent seeking to minimize the total load of its assigned machine.
A schedule is a pure Nash equilibrium (PNE) if no job can migrate to another
machine and experience a strictly lower load.

\begin{lemma}[PNE characterization]\label{lem:pne}
A schedule $S$ is a PNE if and only if for every machine $r$ and every job $g$
assigned to $r$,
\begin{equation}\label{eq:pne}
\ell_r(S)-p_g\le\ell_{\min}(S).
\end{equation}
\end{lemma}
\begin{proof}
For any job assigned to a machine that is not least-loaded, the most attractive
migration destination is a least-loaded machine. Migrating there results in a
new machine load of $\ell_{\min}(S)+p_g$. The migration is unprofitable if and
only if $\ell_r(S)\le\ell_{\min}(S)+p_g$, which matches~\eqref{eq:pne}.
If machine $r$ already achieves the minimum load, the condition is satisfied
automatically and no destination can improve the job's cost.
\end{proof}

For a fixed number of machines $n$, the pure price of anarchy is defined as
\[
\PoA_n^{\mathrm{cov}}
=\sup_P\sup_{S\in\mathrm{PNE}(P),\,\ell_{\min}(S)>0}
\frac{\OPT_n(P)}{\ell_{\min}(S)}.
\]

\subsection{EFX and pure Nash equilibria}

\begin{proposition}[Fixed-$n$ EFX correspondence]\label{thm:efx-equivalence}
For every $n\ge2$,
\begin{equation}\label{eq:efx-equivalence}
\rho_n^{\EFX}=\rho_n^{\EFX_0}=\frac1{\PoA_n^{\mathrm{cov}}}.
\end{equation}
\end{proposition}
\begin{proof}
Let $A$ be an EFX allocation and let $i$ be an agent with positive MMS, so
$s=v_i(A_i)>0$ by Theorem~\ref{thm:efx-direct}. We construct a scheduling
instance by creating a job of size $p_g=v_i(g)/s$ for each good $g$ with
$v_i(g)>0$, mapping bundle assignments to machines. Machine $i$ has load one.
By EFX, for any job $g$ on machine $j\ne i$,
\[
\ell_j-p_g=\frac{v_i(A_j\setminus\{g\})}{s}\le1.
\]
For any machine with load strictly below one, we add a dummy job to bring its
load to exactly one. Machine $i$ maintains load one, every machine now has load
at least one, and removing any job from an augmented machine leaves a load of
at most one. Removing a job from machine $i$ also leaves load at most one.
By Lemma~\ref{lem:pne}, this schedule forms a PNE with minimum load one.
The optimum before padding equals $\MMS_i/s$, and padding cannot decrease it.
The augmented instance thus has $\OPT_n\ge\MMS_i/s$, implying
$\MMS_i/s\le\PoA_n^{\mathrm{cov}}$ and
$\rho_n^{\EFX}\ge1/\PoA_n^{\mathrm{cov}}$.

Conversely, consider any PNE schedule $S$ with $\ell_{\min}(S)>0$ and let $c$ be
a least-loaded machine. We create an agent for each machine and an item for
each job with value $v_c(g)=p_g$ to the evaluating agent. Each peer agent values
each item assigned to her own machine at one and every other item at zero.
By Lemma~\ref{lem:pne}, agent $c$ satisfies EFX$_0$, while all other agents
experience zero envy. The evaluating agent's maximin share is $\OPT_n(P)$ and
her allocated value is $\ell_{\min}(S)$. Hence every PNE covering ratio is
realized as the reciprocal of an EFX$_0$-to-MMS ratio. Taking the supremum of
covering ratios gives $\rho_n^{\EFX_0}\le1/\PoA_n^{\mathrm{cov}}$.
Combined with $\rho_n^{\EFX}\le\rho_n^{\EFX_0}$, this completes the equivalence.
\end{proof}

Combining Proposition~\ref{thm:efx-equivalence} with Corollary~\ref{cor:exact}
yields an independent recovery of the universal covering price of anarchy,
originally established in~\cite{ChenEtAl2013}:
\begin{corollary}[Universal pure price of anarchy]\label{cor:poa}
\[
\sup_{n\ge2}\PoA_n^{\mathrm{cov}}=\frac{17}{10}.
\]
\end{corollary}

\subsection{PMMS and exact pairwise repartition}

For a schedule $S$ and a designated machine $c$, let $P_c$ denote its assigned
jobs and $\ell_c$ its load. To capture this bilateral partition requirement in
scheduling, we introduce the two-machine repartition locality gap $\Gamma_n$,
defined as
\begin{equation}\label{eq:Gamma}
\Gamma_n=\sup_{P,S,c}\frac{\OPT_n(P)}{\ell_c(S)},
\end{equation}
where the supremum is taken over all positive job multisets $P$ with
$\OPT_n(P)>0$, schedules $S$ on $n$ machines, and designated machines $c$ with
$\ell_c(S)>0$ satisfying the pairwise optimality condition
\begin{equation}\label{eq:two-machine-condition}
\OPT_2(P_c\cup P_j)\le\ell_c(S)\qquad\text{for all }j\ne c.
\end{equation}

\begin{proposition}[Fixed-$n$ PMMS correspondence]\label{thm:pmms-equivalence}
For every $n\ge2$,
\begin{equation}\label{eq:pmms-equivalence}
\boxed{\rho_n^{\PMMS}=\frac1{\Gamma_n}.}
\end{equation}
The reduction preserves individual ratios in both directions.
\end{proposition}
\begin{proof}
Given a PMMS allocation and an agent $i$ with positive MMS,
Theorem~\ref{thm:pmms-direct} ensures $v_i(A_i)>0$. Construct a job for each
good with $v_i(g)>0$, with processing time $p_g=v_i(g)$, and discard zero-valued
goods. Condition~\eqref{eq:two-machine-condition} holds directly for $c=i$
from the definition of PMMS. Since $\OPT_n(P)=\MMS_i$ and
$\ell_c=v_i(A_i)$,
\[
\frac{\MMS_i}{v_i(A_i)}\le\Gamma_n,
\qquad
\frac{v_i(A_i)}{\MMS_i}\ge\frac1{\Gamma_n}.
\]
Thus $\rho_n^{\PMMS}\ge1/\Gamma_n$.

Conversely, given any feasible configuration $(P,S,c)$ in the definition of
$\Gamma_n$, define a fair-division instance where agent $c$ evaluates items by
their processing times and each peer agent values her own assigned items at one
and all other items at zero. The allocation satisfies PMMS for all agents and
realizes the focal ratio $\ell_c(S)/\OPT_n(P)$ exactly. Taking infima over these
reciprocal ratios gives $\rho_n^{\PMMS}\le1/\Gamma_n$. This argument uses a
supremum and does not assume that an optimizing configuration exists.
\end{proof}

Machine-covering local search also appears in online algorithms with bounded
migration~\cite{GalvezEtAl2018}. Here exact two-machine repartition gives the
following comparison with Nash equilibria.

\begin{proposition}\label{prop:gamma-poa}
For every $n\ge2$,
\[
\Gamma_n\le\PoA_n^{\mathrm{cov}}.
\]
\end{proposition}
\begin{proof}
Scale the instance so that $\ell_c=1$. If any peer machine $j$ contained a job
$g$ such that $\ell_j-p_g>1$, the bipartition of $P_c\cup P_j$ into
$P_j\setminus\{g\}$ and $P_c\cup\{g\}$ would leave both machines with loads
strictly exceeding one, contradicting~\eqref{eq:two-machine-condition}.
Thus $\ell_j-p_g\le1$ holds for all peer jobs. Padding machines with load below
one creates a valid PNE with minimum load one, and cannot decrease the covering
optimum. Every feasible ratio for $\Gamma_n$ is therefore at most
$\PoA_n^{\mathrm{cov}}$, proving the claim.
\end{proof}

The correspondence and Corollary~\ref{cor:exact} yield the matching lower sequence:
\begin{corollary}[Universal locality gap]\label{cor:gamma}
\[
\sup_{n\ge2}\Gamma_n=\frac{17}{10}.
\]
\end{corollary}

\subsection{Finite-agent consequences}

The fixed-$n$ correspondences immediately transfer exact finite-machine bounds
from the scheduling literature to fair division:
\begin{corollary}\label{cor:small-n}
For $n=2,3$, $\rho_n^{\EFX}=\rho_n^{\EFX_0}=2/3$; for $n=4,5$, both factors
equal $8/13$; and for $n=6,7$, both equal $3/5$. In contrast, for PMMS,
\[
\rho_2^{\PMMS}=1,\qquad\rho_3^{\PMMS}=\frac23.
\]
\end{corollary}
\begin{proof}
The values for $\rho_n^{\EFX}$ and $\rho_n^{\EFX_0}$ follow directly from
Proposition~\ref{thm:efx-equivalence} and the exact prices of anarchy established
in~\cite[author manuscript, Theorem 9]{ChenEtAl2013}.

For PMMS with $n=2$, the definition coincides with MMS, so
$\rho_2^{\PMMS}\ge1$; two equally valued goods allocated one to each agent
attain ratio one. For $n=3$, PMMS implies EFX: a violation of EFX by a positive
good $g\in A_j$ would make both parts of
$(A_j\setminus\{g\},A_i\cup\{g\})$ exceed $v_i(A_i)$, contradicting PMMS
\cite[Theorem 4.6]{CaragiannisEtAl2019}. Hence
$\rho_3^{\PMMS}\ge\rho_3^{\EFX}=2/3$. To show tightness, consider three unit
goods and three goods of value $1/2$, partitioned as $\{1\}$, $\{1,1\}$, and
$\{1/2,1/2,1/2\}$ to the evaluating agent. Give each other agent value only on
her own assigned goods. The evaluating agent with bundle $\{1\}$ satisfies
PMMS by~\eqref{eq:singleton-pmms}, and all peer agents satisfy PMMS. A global
partition into three bundles of type $\{1,1/2\}$ achieves the average $3/2$, so
the MMS is exactly $3/2$, realizing the ratio $2/3$
\cite[Proposition 4.3]{ABM2018}.
\end{proof}

The sharp contrast between $\rho_2^{\PMMS}=1$ and $\rho_2^{\EFX}=2/3$
illuminates the structural difference between the two notions. For two agents,
PMMS explicitly enforces the global maximin share guarantee. EFX, however,
permits an asymmetric allocation such as $A_1=\{1/2,1/2\}$ and $A_2=\{1,1\}$
under identical valuations over four goods. Here agent $1$ receives value one
and envies no single-item removal from $A_2$, yet the optimal partition into two
$\{1,1/2\}$ bundles certifies an MMS of $3/2$, yielding the ratio $2/3$
\cite[Proposition 3.3]{ABM2018}. As shown below, this finite-agent advantage
vanishes asymptotically: both factors converge to $10/17$ as $n\to\infty$.

\subsection{Monotonicity with the number of agents for PMMS}

\begin{proposition}\label{prop:monotone}
The sequence $\Gamma_n$ is nondecreasing in $n$, and consequently
$\rho_n^{\PMMS}$ is nonincreasing. Furthermore,
\[
\lim_{n\to\infty}\rho_n^{\PMMS}=\frac{10}{17}.
\]
\end{proposition}
\begin{proof}
Given any feasible instance for $\Gamma_n$ with designated load $L$ and optimal
cover $C$, add an $(n+1)$-st machine containing a single job of size
$X\ge\max\{L,C\}$. In any bipartition of the designated machine with this new
machine, the part not containing the new job is a subset of the designated jobs,
with load at most $L$. All pairwise constraints remain satisfied, and the old
optimal schedule together with the new job gives a global cover at least $C$.
Taking the supremum over old feasible ratios gives $\Gamma_{n+1}\ge\Gamma_n$.

By the reciprocal identity, $\rho_n^{\PMMS}$ is nonincreasing. It is bounded
below by $10/17$, while Theorem~\ref{thm:tight-family} provides a subsequence
converging to $10/17$, so the entire sequence converges to this limit.
Finally, PMMS implies EFX, and therefore
\[
\frac{10}{17}\le\rho_n^{\EFX}\le\rho_n^{\PMMS}\longrightarrow\frac{10}{17}.
\]
This also justifies the asserted asymptotic limit for EFX.
\end{proof}

\subsection{The historical homogeneous construction}\label{sec:sylvester}

Amanatidis, Birmpas, and Markakis~\cite[Proposition 4.3]{ABM2018} constructed
finite PMMS instances using item values
\[
\frac12,\ \frac16,\ \frac1{42},\ \frac1{1806},\ldots.
\]
We extend this finite pattern using the Sylvester recurrence
\[
d_1=2,\qquad d_{t+1}=d_t(d_t+1),
\]
yielding $d_t\in\{2,6,42,1806,\ldots\}$. The corresponding MMS ratios are
\begin{equation}\label{eq:syl-ratio}
c_r=\left(1+\sum_{t=1}^{r}\frac1{d_t}\right)^{-1}.
\end{equation}
As $r\to\infty$, these ratios decrease to the limit
\[
c_\infty=0.5913554920568897888\ldots.
\]
Although this recurrence provides valid upper bounds, its limit $c_\infty$ lies
strictly above the true universal factor $10/17\approx0.588235$.

For completeness, the finite-depth allocation and its exact MMS can be verified
as follows. For $r\ge1$, take $n=d_{r+1}/2$ unit goods and, for each $t\le r$,
$n$ goods of value $1/d_t$. The recurrence gives $d_t\equiv2\pmod4$, so $n$ is
odd. Moreover, $d_{r+1}$ is divisible by every earlier $d_t+1$; these divisors
are odd, so each $n/(d_t+1)$ is an integer. Give the evaluating agent one unit
good, allocate $n/(d_t+1)$ bundles of $d_t+1$ equal goods at each layer, and
allocate the remaining unit goods in pairs. The telescoping identity
\[
\sum_{t=1}^r\frac1{d_t+1}=\frac12-\frac1{d_{r+1}}
\]
ensures that the total number of allocated bundles is
\[
1+\sum_{t=1}^r\frac{n}{d_t+1}+\frac{n-1}{2}
=1+\left(\frac n2-\frac12\right)+\frac{n-1}{2}=n.
\]
All peer residuals equal one, so~\eqref{eq:singleton-pmms} proves the evaluating
agent's PMMS condition. Give the other agents positive value only on their own
bundles. A benchmark partition with one unit good and one good from every layer
in each part has value $1+\sum_{t=1}^r1/d_t$. This equals the average value,
so the evaluating agent's MMS is exactly this benchmark and her ratio is $c_r$.

The strict gap above $10/17$ also admits an exact rational verification.
Since $d_{t+1}\ge7d_t$ for $t\ge2$,
\[
\sum_{t\ge1}\frac1{d_t}
\le\frac12+\frac{1/6}{1-1/7}=\frac{25}{36}<\frac7{10}.
\]
Consequently $c_\infty\ge36/61>10/17$. In addition, $c_4=301/509$ and
$c_\infty<c_4$, which verifies the comparisons in~\eqref{eq:300509}.

\begin{figure}[t]
\centering
\includegraphics[width=.80\textwidth]{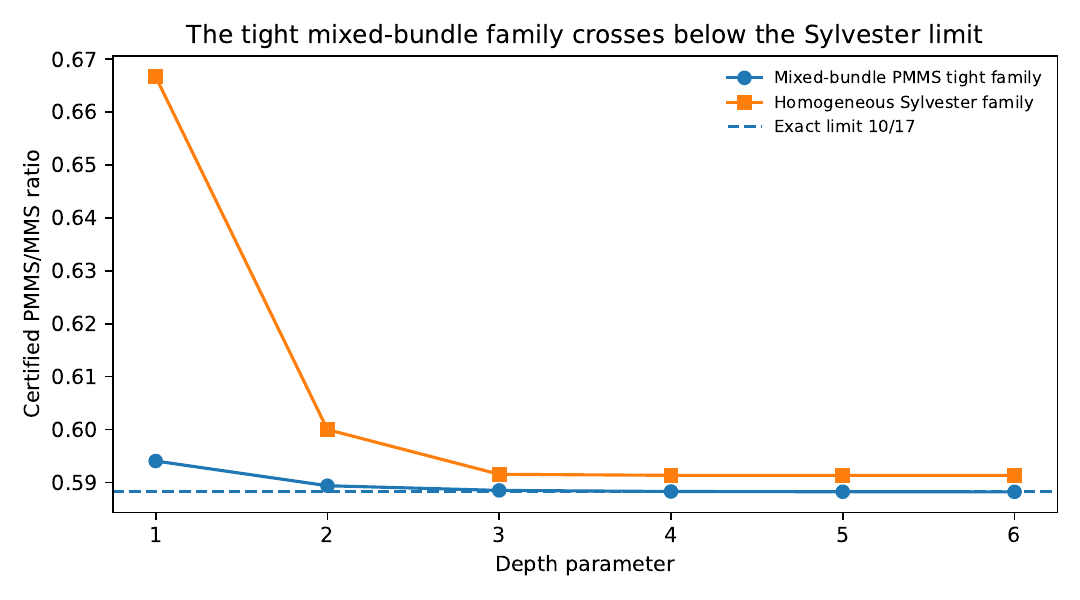}
\caption{Finite-depth certificates $r_\ell$ for the mixed family and exact
ratios $c_r$ for the homogeneous family. The blue values are upper bounds on the
actual ratios, not claimed exact MMS values. The horizontal axis is construction
depth; the numbers of agents in the two families differ.}
\Description{Two sequences of certificates versus depth: the mixed family tends
to ten seventeenths, while the homogeneous Sylvester family tends to
approximately 0.59135549.}
\label{fig:tight}
\end{figure}

We can characterize the precise barrier preventing the historical construction
from reaching $10/17$. Consider the following restricted layered model: each
layer $j$ consists of homogeneous bundles containing $k_j\ge2$ identical goods,
each worth $1/(k_j-1)$, with an equal number of goods across layers. Each layer
consumes a fraction $1/k_j$ of the available bundles while contributing
$1/(k_j-1)$ to an MMS certificate bundle. Finding the best asymptotic certificate
within this structural family reduces to the integer program
\begin{equation}\label{eq:layerproblem}
\max\left\{\sum_{j=1}^r\frac1{k_j-1}:
2\le k_1\le\cdots\le k_r,\quad
\sum_{j=1}^r\frac1{k_j}<1,\quad k_j\in\mathbb Z\right\}.
\end{equation}
The reciprocal-capacity sum $\sum_j1/k_j$ constrains the allocation, whereas the
objective $V=\sum_j1/(k_j-1)$ is the benchmark value.

Every feasible sequence has a valid realization. Choose $n$ divisible by all
$k_j$ and large enough that $n(1-\sum_j1/k_j)\ge1$. For each layer, create $n$
goods worth $1/(k_j-1)$ and group them into $n/k_j$ homogeneous bundles. Add a
separate unit good for the evaluating singleton, and leave any remaining
allocated bundles empty. Every nonempty peer residual equals one, so the
allocation extends to PMMS and EFX$_0$ by giving peer agents value only on their
own items. A benchmark has $n$ parts, each containing one good from every
layer; add the extra unit good to one part. The average-value upper bound then
gives
\[
V\le\MMS_i\le V+\frac1n.
\]
Thus the actual ratios tend to $1/V$ as $n\to\infty$. This justifies the
optimization model, including repeated layer sizes, without claiming to
characterize homogeneous constructions outside this model.

Letting $a_1=1$ and $a_{j+1}=a_j(a_j+1)$, the sequence
$a_j+1\in\{2,3,7,43,\ldots\}$ is the classic Sylvester sequence.

\Needspace{12\baselineskip}
\begin{proposition}[Optimality of the Sylvester sequence within layered designs]\label{thm:layered}
For every feasible sequence in~\eqref{eq:layerproblem},
\[
\sum_{j=1}^r\frac1{k_j-1}\le\sum_{j=1}^r\frac1{a_j},
\]
with equality holding if and only if
\[
(k_1,\ldots,k_r)=(a_1+1,\ldots,a_r+1).
\]
Consequently, the optimal infinite-depth ratio achievable within the specified
layered model is $c_\infty$, which remains strictly greater than $10/17$.
\end{proposition}
\begin{proof}
The recurrence yields the identity $1/(a_j+1)=1/a_j-1/a_{j+1}$, which telescopes to
\begin{equation}\label{eq:syl-prefix}
\sum_{j=1}^{t-1}\frac1{a_j+1}=1-\frac1{a_t}.
\end{equation}
If a feasible nondecreasing sequence equals the Sylvester choice at every index,
it attains equality in the objective. Otherwise, let $t$ be the first index
where it deviates, so $k_t\ne a_t+1$. Feasibility precludes $k_t\le a_t$,
as~\eqref{eq:syl-prefix} would then imply $\sum_{j=1}^t1/k_j\ge1$.
Hence $k_t\ge a_t+2$. If $t=r$, the contribution
$1/(k_t-1)\le1/(a_t+1)<1/a_t$ is strictly smaller. If $t<r$, strict feasibility
gives the tail bound $\sum_{j=t}^r1/k_j<1/a_t$. Because $x/(x-1)$ is strictly
decreasing for $x>1$ and $k_j\ge a_t+2$ for all $j\ge t$,
\begin{align*}
\sum_{j=t}^r\frac1{k_j-1}
&=\sum_{j=t}^r\frac{k_j}{k_j-1}\frac1{k_j}\\
&<\frac{a_t+2}{a_t+1}\frac1{a_t}
=\frac1{a_t}+\frac1{a_{t+1}}
\le\sum_{j=t}^r\frac1{a_j}.
\end{align*}
Thus every deviation from the Sylvester choices strictly degrades the objective.
The Sylvester sequence is itself feasible by the telescoping identity, confirming
that it is the unique optimum for this family.
\end{proof}

\begin{table}[t]
\centering
\small
\begin{tabular}{@{}ccc@{}}
\toprule
Depth $t$ & Mixed certificate $r_t$ & Historical ratio $c_t$\\
\midrule
1 & $60/101\approx0.594059$ & $2/3\approx0.666667$\\
2 & $300/509\approx0.589391$ & $3/5=0.600000$\\
3 & $1260/2141\approx0.588510$ & $42/71\approx0.591549$\\
4 & $5100/8669\approx0.588303$ & $301/509\approx0.591356$\\
$\infty$ & $10/17\approx0.588235$ & $c_\infty\approx0.591355$\\
\bottomrule
\end{tabular}
\caption{Comparison of upper certificates: the mixed-bundle family converges
to $10/17$, whereas the homogeneous Sylvester family approaches $c_\infty$.
Equal depth parameters need not correspond to equal numbers of agents.}
\label{tab:comparison-values}
\end{table}

Proposition~\ref{thm:layered} and Table~\ref{tab:comparison-values} clarify the
mathematical landscape. The historical sum $1+\sum_{t=1}^r1/d_t$ is the layered
objective with $r+1$ layers, because $a_1=1$ and $a_{t+1}=d_t$. The gap between
earlier upper bounds and $10/17$ is not an artifact of insufficient recursion
depth within this model, but an intrinsic limitation of its homogeneous layers.
The mixed construction leaves this class through heterogeneous bundles such as
$\{a_i,b_i,b_i\}$ and $\{c_{i+1},d_i,d_i,d_i,d_i,d_i\}$, which exploit the
$17/10$ packing structure.

\section{Conclusion and Discussion}

We resolved the quantitative relationship between local pairwise fairness and the global maximin share benchmark under nonnegative additive valuations: EFX and PMMS have the same tight universal MMS factor, $10/17$. A common local residual inequality and a piecewise charging function give the lower bound and the extension from $\alpha$-EFX to $(10\alpha/17)$-MMS. A complete seven-block allocation family proves tightness while satisfying both PMMS and EFX$_0$.

The result has a direct multi-agent interpretation. MMS verification requires global reasoning over partitions of all resources, whereas EFX and PMMS use bilateral bundle comparisons. These local checks nevertheless protect at least $58.82\%$ of every agent's MMS. The stronger PMMS condition offers no additional universal protection in the worst case, so a system concerned only with this factor gains nothing by replacing EFX verification with PMMS verification.

The scheduling correspondence explains the constant structurally. For every fixed number of agents, the EFX guarantee is reciprocal to the pure price of anarchy in selfish identical-machine covering, while the PMMS guarantee is reciprocal to an exact two-machine repartition locality gap. Thus $17/10$ describes the same obstruction in resource allocation and load covering.

Open questions include the exact finite-$n$ PMMS factors, extensions beyond additive valuations, and dynamic or online arrival. The implication studied here is also distinct from algorithms combining EF1 with Pareto efficiency \cite{barman2018efficient,CaragiannisEtAl2019}, incentive-aware share domination \cite{BabaioffFeige2025}, and online MMS scheduling \cite{ZhouBaiWu2023}.

\label{end:mainbody}

\clearpage
\begingroup
\small
\bibliographystyle{plain}
\bibliography{references}
\endgroup
\end{document}